\documentclass{article}
\usepackage{graphicx,color}
\usepackage{amsmath}
\usepackage{amsthm}
\usepackage{amssymb}
\usepackage{ulem}
\usepackage{natbib}

\usepackage{url}

\usepackage[margin=1.5in]{geometry}
\usepackage{comment}

\newtheorem{theorem}{Theorem}

\newtheorem{remark}{{\it Remark}}[section]
\newtheorem{coro}{Corollary}

\title{Bias Corrected Variance Stabilizing Transformation for Small Area Estimation}
\author{Masayo Y. Hirose, Malay Ghosh, and Mayumi Oka\\
Kyushu University, University of Florida, Institute of Statistical Mathematics}
\date{\empty}

\begin{document}
\maketitle

\begin{abstract}
Small area estimation models are typically based on the normality assumption of response variables. More recently, attention has been drawn to the transformation of the original variables to justify the assumption of normality. Variance stabilizing transformation of observation serves the dual purpose of reaching closer to normality, as well as known variance of the transformed variables in contrast to the assumption of known variances of the original variables, the latter needed to avoid non-identifiability. 
However, the existing literature on the topic ignores a certain bias introduced in the seemingly correct back transformation. 
The present paper rectifies this deficiency by introducing asymptotically unbiased empirical Bayes (EB) estimators of small area means. Mean squared errors (MSEs) and estimated MSEs of such estimators are provided. The theoretical results were accompanied with simulations and data analysis. A somewhat surprising phenomenon is a finding which connects one of our results to the natural exponential family quadratic variance function (NEF-QVF) family of distributions introduced by Morris (1982,1983). 
\end{abstract}

{\it Keywords:} Confidence Interval, Empirical Bayes estimation, Linear mixed model, Mean squared error estimation, NEF-QVF. 

\section{Introduction}

Small area estimation is gaining increasing popularity in recent years from both public and private sectors. It is now widely recognized that small area estimates need to be based on models linking the different areas. This is because sample sizes in the individual areas are often inadequate to provide adequate precision for direct estimates. 

Small area models, by and large, are linear, starting with the classic paper of \cite{fay1979} targeted to estimate per capita income for small places. 
Their model is essentially a mixed effect model which assumed normality of the errors. Over the years, there have been several extensions of this model, needed to accommodate analysis of specific data at hand. 

One of the issues that received attention a while ago is that the normality assumption of the errors may not always be justified when data are measured in the original scale and need some transformation. \cite{slud2006} initiated this and obtained explicit bias-adjusted estimators after back transformation of data from the logarithmic scale to the original scale. This was followed later for binary data using arcsin transformation by several authors (\cite{raghunathan2007}; \cite{casas2016}; \cite{franco2019}; \cite{hirose2023}; and \cite{hadam2024}). 

Variance stabilizing transformations considered in this paper serve the dual purpose of rendering transformed variables closer to normality as well as getting known sample variances instead of assuming them to be known as is customary to avoid non-identifiability.  

We are also considering the scenario where not just the number of small areas say $m$, is large, but the sample sizes in many of the small areas are also large or at least moderately large. 
This particular feature needs to be recognized when developing expressions for asymptotic mean squared errors and their estimators. 
We consider the asymptotic setting where $n_i$, the sample size in the $i$th small area is $O(N)$ for large $N$ with $N=\sum_i^m n_i$.  

\cite{sugasawa2017} considered dual-power transformation and obtained empirical Bayes estimators of parameters of interest. \cite{hadam2024} proposed empirical Bayes estimators which did not admit closed form solutions and suggested numerical methods for their evaluation.  
They also suggested mean squared error (MSE) estimation employing bootstrap methods while these authors applied back-transformed estimates of parameters of interest. It turns out though that in all the previously cited papers the back transformed estimates involved a certain bias.  This will be made more precise in the next section. The objective of our work is to rectify this deficiency and provide an asymptotically bias-corrected estimate up to a certain order. Moreover, we use survey weighted estimates of these means, as is customary in survey sampling. We introduce empirical Bayes estimates and study their asymptotic behavior, including asymptotically valid empirical Bayes confidence intervals. 

While developing our procedure, we found a surprisingly new result. For variance stabilizing transformation $g$, if we require $g^{\prime\prime}(\mu)/[g^{\prime}(\mu)]^3$ to be linear in $\mu$, (as needed by us), the necessary and sufficient condition for the same is that the population variance is at most a quadratic function of the population mean.  
We also show that such property is achieved when several specific transformations $g(\cdot)$, for example, linear function, square root, arc-sin, logarithm and arc-hyperbolic-sin transformations.

The outline of the remaining sections is as follows. 
In Section 2, we pose the problem of interest and propose the empirical Bayes estimators. In Section 3, we provide asymptotically second-order correct mean square expansion of such estimators, and also second-order corrected estimators of the MSE's. 
In Section 4, we provide second-order correct interval estimates of the small area means. Simulation study and Data analysis are given in Sections 5 and 6. Some final remarks are made in Section 7.

\section{Empirical Bayes Estimation}

Consider $m$ small areas (or domains) labelled $1,\cdots m$. 
Let $y_{ij} (j=1,\cdots n_i)$ denote the $j$-th unit of the ith small area. It is assumed that $E[y_{ij}]=\mu_i$ and $V(y_{ij})=\sigma_i^2$. 
The survey-weighted estimates of the $\mu_i$ (\cite{you2002}) are given by ${y}_i=\sum_{j =1}^{n_i} w_{ij}y_{ij}$ where $\sum_{j =1}^{n_i} w_{ij}=1$ with known $w_{ij}$.  
Thus $E(y_i)=\mu_i$ and $\sigma_{iw}=V(y_i)=\sigma_i^2\sum_{j =1}^{n_i} w_{ij}^2$ $(i=1,\cdots m)$. 
For secondary users of survey, the microdata $y_{ij}$ are not available, and inference needs to be made based on the area level data $y_i$ $(i=1\cdots m)$. 

Suppose now for better approximation to normality, $g$ is a variance stabilizing monotonically strictly increasing function. Following \cite{fay1979}, we introduce the mixed effects model, 
\begin{align}
\label{model}
g({y}_i)=\theta_i+e_i, \theta_i=x_i^{\prime}\beta+u_i \ (i=1,\cdots, m).
\end{align}

In the above, the $e_i$ and $u_i$ are mutually independent with $e_i\sim^{ind} N(0,D_i)$ and $u_i\sim^{iid} N(0,A)$.
To avoid nonidentifiability, the $D_i$ are assumed known. 

The customary approach is to estimate the $\mu_i$ by $g^{-1}(\hat \theta_i)$ or $E[g^{-1}(\theta_i)|g(y_i)]$, where the $\hat \theta_i$ are estimators of the $\theta_i$. For instance, if $g(y_i)=\log (y_i)$, one estimates $\mu_i$ by $\exp[\hat \theta_i]$ or $\exp[\hat \theta_i+\frac{1}{2}\hat V(\hat \theta_i)]$, where $\hat V(\hat \theta_i)$ is certain exact or asymptotic estimator of $V(\hat \theta_i)$. But here comes the bias, while $E[g(y_i)]=\theta_i, \mu_i=E[y_i]$ does not imply $g^{-1}(\theta_i)= \mu_i$. 

We take up this issue, which to our knowledge has not been addressed before, and proceed 
to develop new small area estimates under variance stabilizing transformation of the original scale. 
Assume $\sigma_{iw}^2=E[(y_i-\mu_i)^2]=O(n_i^{-1})$, $E[(y_i-\mu_i)^3]=O(n_i^{-3/2}), E[(y_i-\mu_i)^4]=O(n_i^{-2})$ and so on. Now we use the Taylor series expansion. 
\begin{align}
\label{eq2.25312}
\theta_i=E[g(y_i)|\theta_i]=g(\mu_i)+\frac{\sigma_{iw}^2}{2}g(\mu_i)^{\prime\prime}+O(n_i^{-3/2}). 
\end{align}
Thoughout this paper, we denote $\partial_{\mu_i}^{k} g(\mu_i) (k=1,2,3)$ as $g(\mu_i)^{\prime}, g(\mu_i)^{\prime\prime }$ and $g(\mu_i)^{\prime\prime\prime}$, respectively, and assume these are of order $O(1)$. 

Typically, with the above nonlinear complex function, it is hard to make an adequate back transformation. In order overcome this, we consider an initial a shift transformation, namely $\theta_i=g(\mu_i+\delta_i)$ with $\delta_i=O(n_i^{-1})$ and equate the same to the right hand side of \eqref{eq2.25312} up to the order of $O(n_i^{-1})$. 
Though not exact, a similar idea appears in \cite{raghunathan2007}. Also, with a one step Taylor expansion,

\begin{align}
g(\mu_i+\delta_i)=g(\mu_i)+\delta_ig(\mu_i)^{\prime}+O(n_i^{-2}).\label{eq3.25312}
\end{align}
This leads to the relation 
\begin{align}
\delta_i=\frac{\sigma_{iw}^2}{2}g(\mu_i)^{\prime\prime}[g(\mu_i)^{\prime}]^{-1}+O(n_i^{-3/2}).\label{eq4.25312}
\end{align}

Again using 
\begin{align*}
D_i=&V[g(y_i)|\theta_i]=V[g(\mu_i)+(y_i-\mu_i)g(\mu_i)^{\prime}|\theta_i]+O(n_i^{-3/2}),\\
=&\sigma_{iw}^2\left\{g(\mu_i)^{\prime}\right\}^2+O(n_i^{-3/2}).
\end{align*}

\noindent This leads to the identity 
\begin{align}
\delta_i=&\frac{D_i}{2}g(\mu_i)^{\prime\prime}[g(\mu_i)^{\prime}]^{-3}+O(n_i^{-3/2}).\label{eq5.25312}
\end{align}

\noindent We may also note here that since $g$ is variance stabilizing, $D_i$ does not depend on unknown parameters.

We want to consider the general situation when $\delta_i$ is approximately linear in $\mu_i$ i.e. 
$$\delta_i=\frac{D_{i}}{2}(a_{i} \mu_i+b_{i})+O(n_i^{-3/2}).$$
This leads to the differential equation 
\begin{align}
g^{\prime\prime}(\mu_i)/[g(\mu_i)^{\prime}]^3=a_i\mu_i+b_i,
\label{e1.240703}\end{align}
where $a_{i}$ and $b_{i}$ are some generic constant values of the order $O(1)$.

The above linearity assumption is not as ad-hoc as it sounds. It is clear that the variance stabilizing transformation is based essentially on a mean variance relationship. 
It turns out that for \eqref{e1.240703} to hold, the variance is at most a quadratic function of the mean. 
The following theorem makes it more precise.

\begin{theorem}
    Let $X$ be a random variable with mean $\mu$ and variance $\sigma^2$. 
    Consider a variance stabilizing transformation $g$ satisfying $[g(\mu)^{\prime}]^2\sigma^2=k$, a positive constant.   
    Then if $\sigma^2=c_0+c_1\mu+c_2\mu^2$ where $c_0$, $c_1$ and $c_2$ are not all zeroes, $g^{\prime\prime}(\mu)/[g^{\prime}(\mu)]^3$ is either linear in $\mu$ or is a constant. In particular, $a_i=-c_2/k$ and $b_i=-c_1/(2k)$.
\end{theorem}

\begin{proof}
$$|g^{\prime}(\mu)|=\left\{\frac{k}{c_0+c_1\mu+c_2\mu^2}\right\}^{1/2}.$$
If $$g^{\prime}(\mu)=\left\{\frac{k}{c_0+c_1\mu+c_2\mu^2}\right\}^{1/2},$$
$$g^{\prime\prime}(\mu)=-\frac{k^{1/2}}{2(c_0+c_1\mu+c_2\mu^2)^{3/2}}(c_1+2c_2\mu).$$
Then
$$g^{\prime\prime}(\mu)/[g^{\prime}(\mu)]^3=-\frac{1}{2k}(c_1+2c_2\mu),$$ 
a linear function of $\mu$ or is a constant if $c_2=0$. 
Similarly, if $g^{\prime}(\mu)=-\left\{\frac{k}{c_0+c_1\mu+c_2\mu^2}\right\}^{1/2},$ $$g^{\prime\prime}(\mu)/[g^{\prime}(\mu)]^3=-\frac{1}{2k}(c_1+2c_2\mu),$$ 
\end{proof}

\begin{remark}
\label{remark2.1}
    \cite{morris1982,morris1983} introduced the natural exponential family quadratic variance function (NEF-QVF) family of distributions and characterized distributions where the stated $\mu$, $\sigma^2$ relation is satisfied. The six root distributions are Bernoulli, Posson, negative Binomial, Normal, Gamma and generalized hyperbolic secant (GHS).
    We find expressions for $g^{\prime\prime}(\mu)/[g^{\prime}(\mu)]^3$ for the any of NEF-QVF family of distributions. 
    \begin{description}
    \item [I]Bernoulli ($p$): From $\sigma^2=p(1-p), \mu=p$, we have $c_0=0,c_1=1,c_2=-1$. 
    Then $$g^{\prime\prime}(\mu)/[g^{\prime}(\mu)]^3\propto 2p-1.$$
    \item [II]Poisson ($\lambda$):From $\mu=\lambda$, $c_0=c_2=0,c_1=1$. Then
    $$g^{\prime\prime}(\mu)/[g^{\prime}(\mu)]^3\propto -\frac{1}{2}$$
    \item [III]Negative Binomial: $P(X=s)=\begin{pmatrix}
    s+r-1\\
    s 
    \end{pmatrix} p^r (1-p)^s,s=0,1,2,...$, $r$ known. 
    $$\mu=r(1-p)/p, \sigma^2=r(1-p)/p^2=r\frac{1-p}{p}\{1+(1-p)/p\}=\mu (1+\mu/r).
    $$
    Then $c_0=0, c_1=1, c_2=1/r$, 
    $g^{\prime\prime}(\mu)/[g^{\prime}(\mu)]^3\propto -\frac{1}{2}(1+2\mu/r).$
    \item [IV]Normal (0,1): $c_0=1,c_1=c_2=0,g^{\prime\prime}(\mu)/[g^{\prime}(\mu)]^3=0$. 
    \item [V]Gamma: $f(x)=\exp(-x/\theta)\frac{x^{r-1}}{\theta^r \Gamma(r)}$, $r(>0)$ known;
    $\mu=r\theta, \sigma^2=r\theta^2=\mu^2/r$. 
    Then $c_0=c_1=0,c_2=1/r.$ $g^{\prime\prime}(\mu)/[g^{\prime}(\mu)]^3\propto -\frac{1}{r}\mu$.
    
    \item [VI] Generalized hyperbolic secant (GHS): $f(x)=\exp(\theta x)\cos(\theta)/[2\cosh(\pi x/2)]$.
    $\mu=\tan \theta, \sigma^2=1+\tan^2(\theta), c_0=c_2=1, c_1=0$.
    $g^{\prime\prime}(\mu)/[g^{\prime}(\mu)]^3\propto -\mu$.
    \end{description}
It is true, though, that the result of the theorem goes beyond the NEF-QVF family of distributions. Consider for example $X\sim lognormal (\theta,\phi^2)$, $\phi(>0)$ known.
Then $\mu=\exp(\theta+\phi^2/2), \sigma^2=\mu^2[\exp(\phi^2)-1]=\mu^2c_2$. 
$g^{\prime\prime}(\mu)/[g^{\prime}(\mu)]^3\propto -[\exp(\phi^2)-1]\mu$.
\end{remark}

\begin{remark}
In order to see a potential converse of this theorem, suppose that $$g^{\prime\prime}(\mu)/[g^{\prime}(\mu)]^3=d_1+d_2\mu.$$
On integration, $$-\frac{1}{2[g^{\prime}(\mu)]^2}=\frac{d_0}{2}+d_1 \mu+\frac{d_2}{2}\mu^2.$$
leading to $$[g^{\prime}(\mu)]^2=-1/(d_0+2d_1 \mu+d_2\mu^2).$$
In order that this relation holds for all $\mu$, we must need $d_0+2d_1 \mu+d_2\mu^2<0$ for all $\mu$. 
Thus if $d_1=d_2=0$, we need $d_0<0$. If $d_2=0$, we need $d_0+2d_1 \mu<0$ for all $\mu$. 
Finally, if $d_2\neq 0$, we need $d_0d_2<d_1^2$. 
\end{remark}

When the linear structure holds, 
\begin{align*}
     \theta_i=&g(\mu_i+\delta_i)=g\left(\mu_i+\frac{D_i}{2}(a_i\mu_i+b_i)\right)+O(n_i^{-3/2}),\\
     =&g\left(\mu_i(1+\frac{a_iD_i}{2})+\frac{D_i}{2}b_i\right)+O(n_i^{-3/2}).
\end{align*}
Thus 
\begin{align}
\mu_i=\frac{1}{1+a_iD_i/2}\{g^{-1}(\theta_i)-D_ib_i/2\}+O(n_i^{-3/2}).\label{eq7.25312}
\end{align}

\noindent We also note that the linearity assumption as given in \eqref{e1.240703} also makes solving problems with respect to $\mu_i$ much easier after back transformation. 

Accordingly, if the Bayes estimator of $\mu_i$ is given by 
\begin{align}
    \hat \mu_i^{B}=
    \frac{1}{1+a_iD_i/2}\{\tilde \mu_{i*}^{B}-D_ib_i/2\}+O(n_i^{-3/2}), \label{eq8.25312}
 \end{align}
 where $\tilde \mu_{i*}^{B}=E[g^{-1}(\theta_i)|g(y_i)]$.

\begin{remark} 
\label{remark2.3}

For the Bernoulli example, one has the identity 
$$\sin^{-1}(\sqrt{p})=\frac{\pi}{4}+\frac{1}{2}\sin^{-1}(2p-1).$$ If instead of $\sin^{-1}(\sqrt{p})$, we use $g(p)=\sin^{-1}(2p-1)$ again we get 
$$\frac{g(p_i)^{\prime\prime}}{(g(p_i)^{\prime})^3}=\frac{1}{2}(2p_i-1)\propto 2p_i-1.$$
In this case, $a_i=1, b_i=-1/2$, and $D_i=\sum_j^{n_i}w_{ij}^2+O(n_i^{-3/2})$. Hence, 
$$    \hat \mu_i^{B}=
    \frac{1}{1+\sum_j^{n_i}w_{ij}^2/2}\{\tilde \mu_{i*}^{B}+\sum_j^{n_i}w_{ij}^2/4\}+O(n_i^{-3/2}). 
$$
\end{remark}

In practice, the parameters $\beta$ and $A$ are unknown. So we will use instead the empirical Bayes estimator
\begin{align}
    \hat \mu_i^{EB}=
    \frac{1}{1+a_iD_i/2}\{\tilde \mu_{i*}^{EB}-D_ib_i/2\}+O(n_i^{-3/2}), \label{eq9.25312}
 \end{align}

Moreover, while we consider the asymptotic setting $n_i=O(N)$ for large $N$, this asymptotic setting achieves design consistency of $\hat \mu_i^{EB}$, like the direct estimator $y_{i}$. This is because for large $N$,

\begin{align*}
    \hat \mu_i^{EB}\approx 
    \tilde \mu_{i*}^{B}\approx g^{-1}(g(y_i))=y_i, 
 \end{align*}

\section{Bias and MSE Estimation of Empirical Bayes Estimators}

The linearity assumption \eqref{e1.240703} also yields the following helpful results for approximating and estimating the MSE of the empirical Bayes estimator $\hat \mu_i^{EB}$. 
We first list the regularity conditions needed to establish several theorems.  

\subsection*{The regularity conditions:}
\begin{description}
\item [R1] $rank(X)=p$ is bounded for large $m$;
\item [R2] $x_i^{\prime}(X^{\prime}X)^{-1}x_i=O(m^{-1})$ for large $m$;
\item [R3] The sampling variances $D_i=O(n_i^{-1})$ for large $n_i$, $A\in (0,\infty)$;
\item [R4] The transformed function $g(\cdot)$ is thrice continuous differentiable, and the $k$-th derivative of $g$ are uniformly bounded for $k=1,2,3$.   
\item [R5] The estimator of $A$ satisfies that $E[(\hat A-A)^{j}]=O(m^{-1})$ and $E[(\hat A-A)^{4}]=O(m^{-2})$ for $j=1,2$. 
\end{description}

\noindent We now prove Theorem \ref{mse.app1}. Throughout this paper, we assume that $a_i$ and $b_i$ are known.
\begin{theorem}
\label{mse.app1}
Under the regularity conditions R1-R5, we have, for large $N$, 
\begin{description}
\item [(i)] $E(\hat \mu_i^{EB}-\mu_i)=\frac{1}{1+a_iD_i/2}E(\tilde \mu_{i*}^{EB}-\mu_{i*})+o(n_i^{-1})=O(N^{-1})$,
where $\mu_{i*}=g^{-1}(\theta_i)$. 
\item [(ii)] Let $MSE_i$ denote the MSE of $\hat \mu_i^{EB}$ and let $\tilde \mu_{i*}^{EB}=\tilde \mu_{i*}^{B}(\hat \lambda)$, where $\lambda=(\beta,A)$. $MSE_{i,app*}$ is the second-order approximation of $MSE_{i*}:=E[(\tilde \mu_{i*}^{EB}-\mu_{i*})^2]$, satisfying $MSE_{i,app*}=E[(\tilde \mu_{i*}^{EB}-\mu_{i*})^2]+o(N^{-1})$. Then 
\begin{align*}
MSE_i&:=E[(\hat \mu_{i}^{EB}-\mu_{i})^2]\\
&=\frac{1}{(1+a_iD_i/2)^2}{MSE_{i,app*}}+o(N^{-2})\\
&=\frac{1}{(1+a_iD_i/2)^2}{MSE_{i*}}+o(N^{-1})
\end{align*}
\item [(iii)] Let $\hat M_i$ define $$\hat M_i:=\frac{1}{(1+a_iD_i/2)^2}\hat M_{i*}$$ 
with the second-order unbiased MSE estimator $\hat M_{i*}$ for MSE of $\tilde \mu_{i*}^{EB}$ satisfying $E[\hat M_{i*}-MSE_{i*}]=o(N^{-1})$. Then
\begin{align*}
E[\hat M_i-MSE_i]
&=o(N^{-1}).
\end{align*}
\end{description}
\end{theorem}

\begin{proof}
\noindent {\bf (i)} The equations \eqref{eq8.25312}-\eqref{eq9.25312} yield, under the regularity conditions; 
\begin{align}
E(\hat \mu_i^{EB}-\mu_i)=&E(\hat \mu_i^{EB}-\hat \mu_i^B)+E(\hat \mu_i^{B}-\mu_i),\notag\\
=&\frac{1}{1+a_iD_i/2}E(\tilde \mu_{i*}^{EB}-\tilde \mu_{i*}^B)+o(n_i^{-1}),\label{eq2.250402}\\ 
=&O(m^{-1})+o(n_i^{-1}). \notag
\end{align}

\noindent It holds that $O(n_i^{-1})=O(N^{-1})$ and $O(m^{-1})=O(N^{-1})$ in our asymptotic setting for large $N$. 
This proves (i). 

\noindent {\bf (ii)} 
Using the equation \eqref{eq2.250402}, 
\begin{align*}
MSE_i&=\frac{1}{(1+a_iD_i/2)^2}E[(\tilde \mu_{i*}^{EB}-\mu_{i*})^2]+\frac{2}{1+a_iD_i/2}E(\tilde \mu_{i*}^{EB}-\tilde \mu_{i*}^B)O(n_i^{-3/2})+O(n_i^{-3}),\\
&=\frac{1}{(1+a_iD_i/2)^2}E[(\tilde \mu_{i*}^{EB}-\mu_{i*})^2]+O(m^{-1}n_i^{-3/2})+O(n_i^{-3}),\\
&=\frac{1}{(1+a_iD_i/2)^2}{MSE_{i,app*}}+o(N^{-2}),\\
&=\frac{1}{(1+a_iD_i/2)^2}{MSE_{i*}}+o(N^{-1}).
\end{align*}

\noindent {\bf (iii)} 
From (ii) and the definitions of $\hat M_i$ and $MSE_{i,app*}$, we have, for large $N$, 

\begin{align*}
E[\hat M_i-MSE_i]&=\frac{1}{(1+a_iD_i/2)^2}E[\hat M_{i*}-MSE_{i,app*}]+o(N^{-2}), \\
&=\frac{1}{(1+a_iD_i/2)^2}E[\hat M_{i*}-MSE_{i*}]+o(N^{-1}),\\ 
&=o(N^{-1}). 
\end{align*}
\end{proof}

\noindent Note that in this study, we define the second-order unbiased MSE estimator $\hat{M}_i$ for large $N$ satisfying 
$$E[\hat {M}_i-MSE_i]=o(N^{-1}).$$
Incidentally, the plugged-in estimator of $MSE_{i,app*}$ suffices as $\hat M_i$ in our asymptotic setting where $n_i=O(N)$ because $MSE_{i,app*}=O(N^{-1})$ often holds. Thus, any bias correction is not needed in our asymptotic setting. 

Nevertheless, if $MSE_{i,app*}$ and $\hat M_{i*}$ satisfy that $MSE_{i,app*}=E[(\tilde \mu_{i*}^{EB}-\mu_{i*})^2]+o(N^{-2})$ and $E[\hat M_{i*}-MSE_{i,app*}]=o(N^{-2})$, then $\hat M_{i}$ gets more efficiency for large $N$. 
Luckily, several MSE estimators $\hat M_i^{*}$ may be helpful, which have already been developed for several specific transformations $g(\cdot)$ (\cite{slud2006}; \cite{ghosh2022}; \cite{hirose2023}). 

\section{Confidence Interval}

A confidence interval is also essential for predicting the quantity of interest $\mu_i$. 
One may use the naive method as follows:  
\begin{align}
\label{e1.250601}
I^{naive}_i=[y_i\pm z_{\alpha/2} \sqrt{\hat \sigma_{iw}}]
\end{align} 
However, with the above, there is no guarantee of achieving second-order corrected confidence interval. Also, situation gets more complex when $\sigma_{i}$ is related to $\mu_i$ as pointed out earlier. 
The second-order corrected confidence interval has been developed in small-area estimation, which achieves a nominal coverage probability up to the order of $O(m^{-1})$ for the untransformed case ($g$ being the identity function) with the asymptotic setting for large $m$ with fixed $n_i$. 
There exist certain methods for creating such intervals for untransformed cases (\cite{datta2002}; \cite{sasase2005}; \cite{hall2006}; \cite{chatterjee2008}; \cite{yoshimori2014second}; \cite{diao2014}; 7\cite{hirose2017}).
\cite{casas2016} and \cite{hadam2024} considered the transformed confidence intervals for arc-sin transformed data $\sin^{-1}(\sqrt{y_i})$. \cite{franco2019} investigated the confidence interval for proportions in complex sample surveys, including the arc-sin root transformed data. 
For example, the following interval $I^{CCL}_i$ was mentioned in \cite{casas2016}. 

\begin{align}
\label{CCL}
I^{CCL}_i=\Big[\sin^2(\hat \theta_i^{EB}+q_{i;\alpha/2} \sqrt{\hat g_{1i}}),\sin^2(\hat \theta_i^{EB}+q_{i;1-\alpha/2} \sqrt{\hat g_{1i}})\Big],
\end{align}
where $q_{i;\alpha/2}$ and $q_{i;1-\alpha/2}$ are the $\alpha/2$ and $1-\alpha/2$ quantiles of the bootstrap approximating distribution of the root $(\theta_i-\hat \theta_i^{EB})/\sqrt{g_{1i}(\hat A)}$ given in \cite{chatterjee2008}, respectively, where $\hat \theta_i^{EB}$ is the empirical best linear unbiased predictor (EBLUP) of $\theta_i$ and $g_{1i}(A)=AD_i/(A+D_i)$.  
Additionally, \cite{hadam2024} proposed the percentile parametric bootstrap method. However, they did not provide any theoretical justification.

Unfortunately, these methods have been constructed under the model $\mu_i=g_i^{-1}(\theta_i)$. As far as we know, no one considered the second-order corrected confidence interval under the one proposed in \eqref{eq7.25312} at least for small area estimation. 
Hereafter, we redefine the desired second-order corrected confidence interval $\hat I_i$ for large $N$, such as 
\begin{align*}
    P(\mu_i \in \hat I_i)=1-\alpha+O(N^{-3/2}). 
\end{align*}

\noindent Therefore, this section seeks such a confidence interval to achieve that desired property under \eqref{eq7.25312}.  

To this end, we now get back to the existing confidence intervals $\hat I_i$ for untransformed data such that $P(\theta_i \in \hat I_i)=1-\alpha+O(N^{-3/2})$. We denote the direct and the empirical Bayes confidence intervals as $\hat I_{i}^{D}=g(y_i)\pm z_{\alpha/2}\sqrt{D_i}$ and $\hat I_i^{EB}=[\hat \theta_i^{EB}+q_{i;\alpha/2} \sqrt{\hat M_{i}},\hat \theta_i^{EB}+q_{i;1-\alpha/2} \sqrt{\hat M_{i}}]$, respectively, with some quantiles $q_{i;\alpha/2}$ and $q_{i;1-\alpha/2}$, and some estimators $\hat M_i$ of uncertainty of $\hat \theta_{EB}$ which have developed so far for untransformed data.  
We then establish the following theorem and corollary: 

\begin{theorem}
\label{th3}
Under the regularity conditions R1-R5, we have for large $N$,
    $$P\left(\mu_i\in \frac{1}{1+a_iD_i/2}\{g^{-1}(\hat I_{i})-D_ib_i/2\}\right)=1-\alpha+O(N^{-3/2}).$$
\end{theorem}

\begin{proof}
Using \eqref{eq7.25312},
\begin{align*}
P(\theta_i \in \hat I_{i})=&P\left((1+a_iD_i/2)\mu_i+D_ib_i/2 \in g^{-1}(\hat I_{i})+O(n_i^{-3/2})\right),\\
=&P\left(\mu_i\in \frac{1}{1+a_iD_i/2}\{g^{-1}(\hat I_{i})-D_ib_i/2\}\right)+O(n_i^{-3/2}).
\end{align*}
The result is used in the above second equality that $\sup_{x\in \Xi} f(x)=O(1)$, where $f(x)$ is the density and $\Xi=(\{g^{-1}(\hat I_{i})-D_ib_i/2\}\pm O(n_i^{-3/2}))$. 
Then we consider $\hat I_{i}$ such that $P(\theta_i \in \hat I_{i})=1-\alpha+O(N^{-3/2})$ in our asymptotic setting. 
This completes the proof. 
\end{proof}

\noindent The above theorem provides the following corollary: 
\begin{coro}
    Under the regularity conditions R1-R5, for large $N$,
    \begin{description}
        \item[(i)]$P\left(\mu_i\in \frac{1}{1+a_iD_i/2}\{g^{-1}(\hat I_{i}^D)-D_ib_i/2\}\right)=1-\alpha+O(N^{-3/2}),$
        \item[(ii)]$P\left(\mu_i\in \frac{1}{1+a_iD_i/2}\{g^{-1}(\hat I_{i}^{EB})-D_ib_i/2\}\right)=1-\alpha+O(N^{-3/2}).$
    \end{description}
\end{coro}

\noindent From the above corollary, we call $\hat I_i^{TD}$ and $\hat I_i^{TEB}$ new transformed confidence interval which denoted as $$\hat I_{i}^{TD}:= \frac{1}{1+a_iD_i/2}\{g^{-1}(\hat I_{i}^{D})-D_ib_i/2\}$$ and $$\hat I_i^{TEB}:= \frac{1}{1+a_iD_i/2}\{g^{-1}(\hat I_{i}^{EB})-D_ib_i/2\}.$$  

A pertinent question is which interval is more desirable in terms of length. This question is indeed important especially for small area estimation. 
For untransformed data, $\hat I_i^{TEB}=\hat I_i^{EB}$ has a smaller length than that of $\hat I_{i}^{TD}=\hat I_{i}^{D}$ in the asymptotic sense for large $m$ and fixed $n_i$. In particular, the empirical Bayes confidence interval proposed by \cite{yoshimori2014second} always yields a shorter length than $\hat I_{i}^{D}$, while maintaining the desired asymptotic coverage probability. 
Incidentally, \cite{cox1975} interval also provides this desired coverage probability in our asymptotic setting for large $N$ because $g_{1i}(A)=O(N^{-1})$ and Theorem 1 in \cite{yoshimori2014second}. 

\begin{remark}
Our interval provides a shorter length than the corresponding naive transformed confidence interval due to the effect of the multiplier $\frac{1}{1+a_iD_i/2}$ if $a_i>0$. 
\end{remark}

However, unfortunately, we do not have a guarantee that the $\hat I_i^{TEB}$ has a smaller length than that of $\hat I_{i}^{TD}$ in our asymptotic sense, due to the complex non-linear $g(\cdot)$. 
Nonetheless, we can obtain asymptotic shorter length confidence interval asymptotically than that of $\hat I_i^{TD}$ from the following theorem. 

\begin{theorem}
\label{length.thm}
Under the regularity conditions R1-R5, we have for large $N$, 

\begin{align*}
L_{i;TD}-L_{i;TEB}=&2z_{\alpha/2} \Big[\left(\sqrt{D_{i}}-\sqrt{g_{1i}(\hat A)}\right) g^{-1}(x)^{\prime}\Big |_{x=\hat \theta_i^{EB}}\\
&+\sqrt{D_i}\{g(y_i)-\hat\theta_i^{EB}\}g^{-1}(x)^{\prime\prime}\Big |_{x=\hat \theta_i^{EB}}\Big]
+O_p(n_i^{-2}),\\
=&2z_{\alpha/2} \Big[\left(\sqrt{D_{i}}-\sqrt{g_{1i}(\hat A)}\right) g^{-1}(x)^{\prime}\Big |_{x=\hat \theta_i^{EB}}\\
&+\sqrt{D_i}\{g(y_i)-\hat\theta_i^{EB}\}g^{-1}(x)^{\prime\prime}\Big |_{x=\hat \theta_i^{EB}}\Big]
+O_p(N^{-2}),\\
=&O_p(N^{-3/2})
\end{align*}  
where $L_{i;TD}$ and $L_{i;TEB}$ are the lengths of $\hat I_{i}^{TD}$ and $\hat I_i^{TEB}$, respectively. 
\end{theorem}
\noindent The proof is deferred in Appendix A. 
The above result implies the length of $\hat I_{i}^{TEB}$ provides a smaller length than the direct-based confidence interval $\hat I_{i}^{TD}$, up to the order of $O_p(N^{-3/2})$ for large $N$ when the two terms being of the order $O_p(N^{-3/2})$ on the right sides are positive. 
From the theorem, the below corollary is also obtained as one example.  
\begin{coro}
    Under the regularity conditions, if $g(\cdot)$ satisfies the following conditions [(i) and (iia)] or [(i) and (iib)]: 
\begin{description}
    \item[(i)]$g^{-1}(\cdot)$ is a monotone increasing function of $\mu_i$;
    \item [(iia)]$g^{-1}(\cdot)$ is a convex at the point $\hat \theta_i^{EB}$ and $g(y_i)>\hat \theta_{i}^{EB}$;
    \item [(iib)]$g^{-1}(\cdot)$ is a concave at the point $\hat \theta_i^{EB}$ and $g(y_i)<\hat \theta_{i}^{EB}$;
\end{description}
then the shorter length of $\hat I_i^{TEB}$ is achieved, up to the order of $O_p(N^{-3/2})$. 
\end{coro}
\noindent Note that $P(g_{1i}(\hat A)<D_i)=1$ and $z_{\alpha/2}>0$ with $\alpha<1/2$.
For instance, let us consider example in Remark \ref{remark2.3} with $g^{-1}(\cdot)=\frac{1}{2}(1+\sin(\cdot))$. This is a monotone increasing and convex (or concave) in the case $(\cdot)>0$ ($(\cdot)<0$).  

\begin{remark}
\label{Re1}
If the distribution of the direct estimator $y_i$ in an original scale is discontinuous, we may suggest to add Yates correction for each confidence interval. For example, in the arcsin transformed model, we may add $\pm w_{i;med}/2$ with the median of $w_{ij}$ for i-th area, $w_{i;med}$, for each side of the confidence limit such that: 
\begin{align}
    \hat I_{i}^{C}=\left[\hat I_{i,L}-\frac{w_{i;med}}{2},\hat I_{i,L}+\frac{w_{i;med}}{2}\right],
\end{align}
where $\hat I_{i,L}$ and $\hat I_{i,U}$ are the lower and upper confidence limits, respectively. 
\end{remark}
\noindent In the above, note that $\sum_{j}w_{i;med}y_{ij}/w_{i;med}\sim Bin (n_i,p_i)$ while $y_{ij}\sim^{ind.}Ber (p_i)$ and $1/2$ is often used for the correction for binomial distribution. 

Here is an example of the arc-sin transformed case given in Remark \ref{remark2.3}. 
In this case, we obtain the corrected direct based confidence interval $\hat I_{i}^{CTD}$ as follows: 
$$\hat I_{i}^{CTD}:= \left[\frac{1+\frac{1}{1+D_i/2}\sin(\hat I_{i,L}^{D})}{2}-\frac{w_{i;med}}{2},\frac{1+\frac{1}{1+D_i/2}\sin(\hat I_{i,U}^{D})}{2}+\frac{w_{i;med}}{2}\right].$$

\noindent Similarly, we may also suggest the corrected versions $\hat I_{i}^{CTEB}$ as follows: 
$$\hat I_{i}^{CTEB}:= \left[\frac{1+\frac{1}{1+D_i/2}\sin(\hat I_{i,L}^{EB})}{2}-\frac{w_{i;med}}{2},\frac{1+\frac{1}{1+D_i/2}\sin(\hat I_{i,U}^{EB})}{2}+\frac{w_{i;med}}{2}\right].$$

\section{Simulation Study}

    We conducted a finite simulation study using arcsin-transformed data $g({y}_i)=\sin^{-1}(2y_i-1)$, as described in Remark \ref{remark2.3}. The original-scale data was randomly generated from 
    $y_{ij}\sim^{indep.}Bernoulli(p_i)$, and the direct estimator was considered as $y_i=\sum_j^{n_i} w_{ij}y_{ij}$, with the assumption $\theta_i=\mu+u_i$ where $\mu=0$ and $u_i\sim^{iid}N(0,A)$. 
    We assumed $p_i=\frac{1}{2}\{1+\sin(\theta_i)/(1+D_i/2)\}$ based on \eqref{eq7.25312}, reflecting the implicit relationship between $p_i$ and $\theta_i$ in this simulation study. 
     
This simulation setting considered different survey weight with $w_{ij}=\tilde w_{ij}/\sum_j \tilde w_{ij}$, where $\tilde w_{ij}=\tilde w_{0ij}/\sum \tilde w_{0ij}$ with the pattern $\tilde w_{0ij} \in \{1,1,2,3,3\}$ for each of the five groups of individuals within each area $i$. Additionally we examined scenarios with  $m \in \{15,50\}, n_i \in \{10,100\}$ and $A=0.006$, while setting the replication number to $5,000$.  

\subsection{Evaluation of EB estimators and MSE estimator}
We first evaluated the performance of six empirical Bayes (EB) estimators and the direct estimator, as follows:
\begin{description}
\item [(NBT.RE)]Natural back transformation with the REML estimator of $A$,  
\item [(NBT.YL)]Natural back transformation with the adjusted REML estimator of $A$ suggested by \cite{yoshimori2014} to avoid zero estimates.
\item [(pEB.RE)] EB estimator with the REML estimator of $A$ under the assumption $\theta_i=g(\mu_i)$,
\item [(pEB.YL)] EB estimator with the adjusted REML estimator of $A$ under the assumption $\theta_i=g(\mu_i)$,
\item [(EB.RE)] EB estimator with the REML estimator of $A$ under our new assumption based on \eqref{eq7.25312}, up to the order $O(n_i^{-1})$,
\item [(EB.YL)] EB estimator with the adjusted REML estimator of $A$ under our new assumption based on \eqref{eq7.25312}, up to the order $O(n_i^{-1})$,
\item[(Direct)] the direct estimator $y_i$. 
\end{description}

Table \ref{bias} presents the average absolute bias values (multiplied by $10^2$) across $m$ areas for each estimator. 
Our EB.RE and EB.YL estimators showed slight improvements, particularly in cases where $n_i = 10$.
Furthermore, not only in terms of bias but also in terms of MSE, our EB estimators appeared to outperform the others, as shown in Table \ref{MSE}, where the MSE values are multiplied by $10^4$.

\begin{table}[h!]
\centering
\begin{tabular}{rrrrrrrr}
  \hline
($m,n_i$) & NBT.RE & NBT.YL & pEB.RE & pEB.YL & EB.RE & EB.YL & Direct \\ 
  \hline
(15,10) & 0.069 & 0.072 & 0.068 & 0.071 & 0.067 & 0.069 & 0.170 \\ 
  (15,100) & 0.043 & 0.042 & 0.043 & 0.042 & 0.042 & 0.042 & 0.075 \\ 
  \hline
(50,10) & 0.063 & 0.063 & 0.062 & 0.062 & 0.059 & 0.059 & 0.212 \\ 
  (50,100) & 0.041 & 0.041 & 0.041 & 0.041 & 0.041 & 0.041 & 0.073 \\ 

   \hline
\end{tabular}
\caption{Average absolute bias multiplied by $10^2$ across $m$ areas for each estimator} 
\label{bias}
\end{table}

\begin{table}[h!]
\centering
\begin{tabular}{rrrrrrrr}
  \hline
($m,n_i$) & NBT.RE & NBT.YL & pEB.RE & pEB.YL & EB.RE & EB.YL & Direct \\ 
  \hline
(15,10) & 64.11 & 64.35 & 61.84 & 62.00 & 55.91 & 56.03 & 306.98 \\ 
  (15,100) & 13.26 & 13.09 & 13.23 & 13.06 & 13.18 & 13.01 & 30.85 \\ 
  \hline
(50,10) & 37.25 & 37.23 & 36.24 & 36.21 & 33.30 & 33.27 & 307.38 \\ 
  (50,100) & 11.20 & 11.19 & 11.20 & 11.18 & 11.18 & 11.16 & 30.67 \\ 

   \hline
\end{tabular}
\caption{MSE values multiplied by $10^4$ across $m$ areas for each estimator} 
\label{MSE}
\end{table}

Next, we evaluated the performance of five MSE estimators for our EB estimators: EB.RE and EB.YL, as follows:
\begin{description} 
\item [(M1.RE)] The MSE estimator for EB.RE, constructed using only the first term $\hat{M}_{1i}$ as $\hat{M}_{i*}$ provided in \cite{hirose2023},
\item [(M1.YL)] The MSE estimator for EB.YL, constructed using only the first term $\hat{M}_{1i}$ as $\hat{M}_{i*}$ provided in \cite{hirose2023},
\item [(Ms.RE)] The MSE estimator for EB.RE, incorporating bias-corrected terms as $\hat{M}_{i*}$ presented in \cite{hirose2023},
\item [(Ms.YL)] The MSE estimator for EB.YL, incorporating bias-corrected terms as $\hat{M}_{i*}$ presented in \cite{hirose2023},
\item [(pMs.YL)] The MSE estimator for EB.YL, incorporating bias-corrected terms as $\hat{M}_{i*}$ presented in \cite{hirose2023}, but excluding the term $\frac{1}{(1+a_iD_i/2)^2}$. 
\end{description}

The following PRB and PRRMSE were used for the evaluation:
\begin{align*}
PRB&=\frac{1}{mR}\sum_{i=1}^m \sum_{r=1}^R \frac{\hat M_i^{(r)}-M_i}{M_i}  \times 100,\\
PRRMSE&=\frac{1}{m}\sum_{i=1}^m \frac{1}{RM_i}[\sum_{r=1}^R (\hat M_i^{(r)}-M_i)^2]^{1/2}  \times 100,
\end{align*}
where the replication number was $R = 5,000$, and $M_i$ represented the simulated MSEs of EB.RE (for M1.RE and Ms.RE) and EB.YL (for the other three estimators).
Tables \ref{PRB}-\ref{PRRMSE} indicated that pMs.YL is generally inferior to the other four MSE estimators in terms of PRB and PRRMSE except the case where $m=50$ and $n_i=100$. Additionally, M1.RE and M1.YL appeared to outperform Ms.RE and Ms.YL in cases where $n_i$ is small.

\begin{table}[h!]
\centering
\begin{tabular}{rrrrrr}
  \hline
($m,n_i$) & M1.RE & M1.YL & Ms.RE & Ms.YL & pMs.YL \\ 
  \hline
(15,10) & -20.99 & -11.19 & 87.93 & 94.79 & 119.54 \\ 
  (15,100) & -32.98 & -30.02 & 11.40 & 13.57 & 14.97 \\ 
   \hline
   (50,10) & 36.33 & 38.58 & 92.04 & 94.10 & 118.76 \\ 
  (50,100) & -16.44 & -16.23 & -1.32 & -1.14 & 0.09 \\ 
  \hline
\end{tabular}
\caption{PRB values for each case} 
\label{PRB}
\end{table}

\begin{table}[h!]
\centering
\begin{tabular}{rrrrrr}
  \hline
($m,n_i$) & M1.RE & M1.YL & Ms.RE & Ms.YL & pMs.YL \\ 
  \hline
(15,10) & 83.94 & 74.53 & 103.20 & 106.53 & 131.49 \\ 
  (15,100) & 57.80 & 53.31 & 36.36 & 34.78 & 35.71 \\ 
   \hline
(50,10) & 114.15 & 112.49 & 134.23 & 133.98 & 160.18 \\ 
  (50,100) & 40.33 & 39.87 & 33.85 & 33.47 & 33.87 \\ 
   \hline

\end{tabular}
\caption{PRRMSE values for each case}
\label{PRRMSE}
\end{table}

We further investigated the percentage of REML estimates being zero, as presented in Table \ref{nega.M}. 
In practice, an estimate of $\hat A$ equal to zero is an unrealistic inference, affecting not only the estimation of $A$ but also the mean squared error (MSE) of the empirical Bayes estimator. Notably, the leading term of the MSE estimate gets zero when $\hat A$ is zero. For further details on this issue, refer to \cite{yoshimori2014}. This table shows that the percentage tends to be high when $m$ and $n_i$ are small. As a consequence of the result, the same percentage of the MSE estimate M1.RE being zero was observed in each case.

\begin{table}[h!]
\centering
\begin{tabular}{rrr}
  \hline
$n_i$ & $m=$15 & $m=$50 \\ 
  \hline
10 & 36.82 &20.96 \\ 
100 & 18.48 &3.66  \\ 
   \hline
\end{tabular}
\caption{Percentage (\%) of REML estimates being zero for each case}
\label{nega.M}
\end{table}

\subsection{Evaluation of confidence intervals}

Moreover, we evaluated the performance of six confidence intervals, as follows:
\begin{description} 
\item [(TDirect)] Our newly proposed transformed confidence interval (CI) based on the direct method,
\item [(TEB.YL)] Our newly proposed CI based on the explicit EB confidence interval from \cite{yoshimori2014second}, 
\item [(Boot)] The naive bootstrap CI, constructed using EB.YL and M1.YL, 
\item [(TEB.B)] The bootstrap implicit CI based on Equation \eqref{CCL}, using Li and Lahiri's adjusted REML estimator (\cite{li2010adjusted}), which accounts for the multiplier $\frac{1}{(1 + a_i D_i / 2)}$, 
\item [(pTEB.B)] Our newly proposed bootstrap implicit CI incorporating Li and Lahiri's adjusted REML estimator, which does not account for the multiplier $\frac{1}{(1 + a_i D_i / 2)}$,
\item [(Mpnaive)] The CI based on Equation \eqref{e1.250601}. 
\end{description}

In this simulation study, the bootstrap sample size was set to $10^3$.
Table \ref{CIcp} presented the coverage probability of the confidence intervals (CIs) for nominal coverage 95\%, demonstrating the superiority of most CIs, except for TDirect and Mpnaive that exhibited significant under-coverage issues. The TEB.YL showed over-coverage in all cases.
The average lengths of the four CIs were generally shorter than those of Direct and Mpnaive from Table \ref{CIL}. Notably, our TEB.B had the shortest average length.

\begin{table}[h!]
\centering
\begin{tabular}{rrrrrrr}
  \hline
($m,n_i$) & TDirect & TEB.YL & Boot & TEB.B & pTEB.B & Mpnaive \\ 
  \hline
(15,10) & 92.99 & 96.57 & 95.25 & 97.24 & 97.31 & 89.03 \\ 
  (15,100) & 94.68 & 96.41 & 92.19 & 95.36 & 95.42 & 94.29 \\ 
   \hline
   (50,10) & 93.02 & 98.79 & 95.00 & 98.64 & 98.76 & 89.15 \\ 
  (50,100) & 94.84 & 95.88 & 97.10 & 94.59 & 94.71 & 94.45 \\ 
   \hline

\end{tabular}
\caption{Coverage probability of the CIs for nominal coverage 95\%}
\label{CIcp}
\end{table}

\begin{table}[h!]
\centering
\begin{tabular}{rrrrrrr}
  \hline
($m,n_i$) & TDirect & TEB.YL & Boot & TEB.B & pTEB.B & Mpnaive \\ 
  \hline
(15,10) & 0.56 & 0.44 & 0.46 & 0.39 & 0.42 & 0.62 \\ 
  (15,100) & 0.21 & 0.17 & 0.20 & 0.15 & 0.15 & 0.22 \\ 
     \hline
  (50,10) & 0.55 & 0.36 & 0.53 & 0.33 & 0.35 & 0.62 \\ 
  (50,100) & 0.21 & 0.14 & 0.20 & 0.13 & 0.14 & 0.22 \\ 
   \hline
\end{tabular}
\caption{Average lengths of the CIs for each case} 
\label{CIL}
\end{table}

Furthermore, Table \ref{CI.com} presented the percentage of cases in which the confidence interval length exceeds that of our transformed direct confidence interval(TDirect). According to the table, the naive bootstrap and Mpnaive confidence intervals exhibited high percentages.

\begin{table}[h!]
\centering
\begin{tabular}{rrrrrr}
  \hline
($m,n_i$) & TEB.YL & Boot & TEB.B & pTEB.B & Mpnaive \\ 
  \hline
(15,10) & 0.00133 & 0.00533 & 0 & 0 & 99.99733 \\ 
  (15,100) & 0 & 0.40133 & 0 & 0 & 100.00000 \\ 
   \hline
(50,10)  & 0 & 0.14000 & 0 & 0 & 99.99720 \\ 
(50,100) & 0 & 0.57800 & 0 & 0 & 100.00000 \\ 
   \hline
\end{tabular}
\caption{Percentage (\%) of the larger length than that of TDirect CI} 
\label{CI.com}
\end{table}

\section{Data Analysis: Poverty mapping for each Japanese prefecture}

   Poverty has been recognized as a significant social issue in Japan, particularly over the past decade. To address this challenge, creating a reliable document that provides insights into poverty situations at a small-domain level is essential. \cite{tomuro2016} estimated poverty rates for each prefecture; however, this study relied solely on direct estimates. 
    In this section, we estimated the poverty rate for each of the 47 prefectures, considering data categorized by gender (2) and one age group (25–34) within single-person households (over 15 years old) using our proposed methods. 
    For this purpose, we utilized the employment status survey data from the 2017 official Japanese microdata and the poverty line defined by the minimum cost of living under public assistance, as in \cite{tomuro2016}. Such subsequent data was provided by the National Survey on Public Assistance Recipients. 
    For females, the AIC selected two auxiliary variables: (i) the graduation rate from higher education institutions for each prefecture, based on the 2010 Census data, and (ii) the number of single-person households for each prefecture, based on the 2015 Census data. Additionally, we obtained mean squared error (MSE) estimates and confidence intervals to compare our results with existing methods. The sample sizes range from 28 to 348 for females, and from 53 to 444 for males.

\subsection{Results}
Figure \ref{map} presented the resulting poverty map for each prefecture. The two top sub-figures displayed the direct (left) and EB (right) estimates for males, while the bottom sub-figures showed the direct (left) and EB (right) estimates for females.
Visually, differences between the two estimation methods are noticeable.  Particularly for females, the contrast appeared even more pronounced. Overall, the right sub-figures provided smoother estimates compared to their corresponding left maps.

\begin{figure}[h!]
    \centering
    \includegraphics[width=10cm]{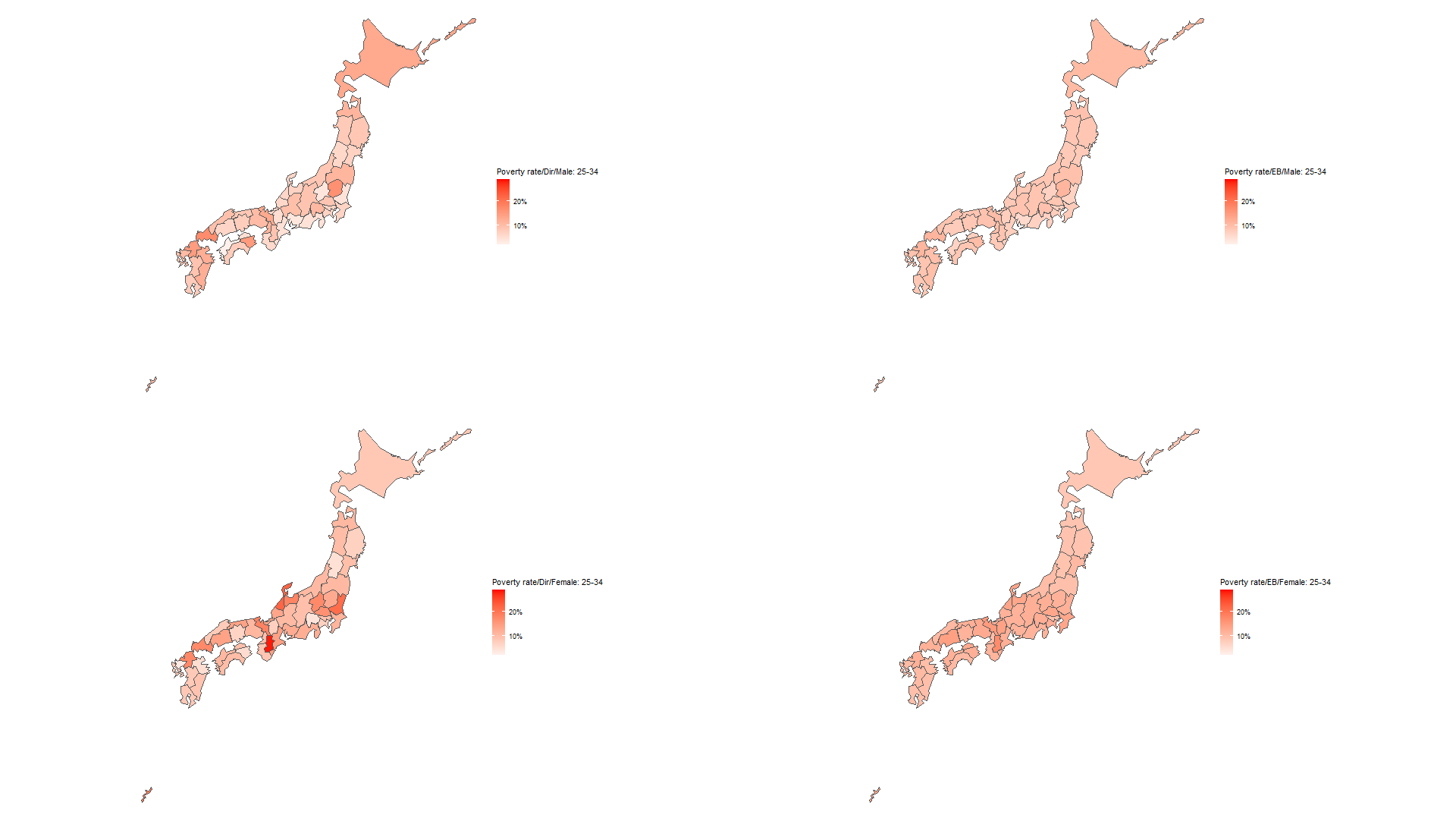}
    \caption{Poverty mapping for each prefecture (top left:male $\times$ direct estimates; top right: male$\times$ EB estimates;bottom left: Female$\times$  direct; bottom right: Female$\times$  EB)}
    \label{map}
\end{figure}

For more details, the four graphs in Figure \ref{Point_est} illustrated poverty estimates for direct, pEB.YL, and EB.YL in the top graphs, while the square root of MSE estimates (direct, M1.YL, Ms.YL, and pMs.YL) are displayed in the bottom section. These figures were arranged in descending order of sample size for each gender, male and female.
From the figures for females, we observed greater differences between the direct and two EB estimates, particularly as the sample size decreases, compared to the results for males. Notably, the discrepancy was especially pronounced in cases with smaller sample sizes.

\begin{figure}[h!]
    \centering
    \includegraphics[width=10cm,bb=0 0 1920 1082]{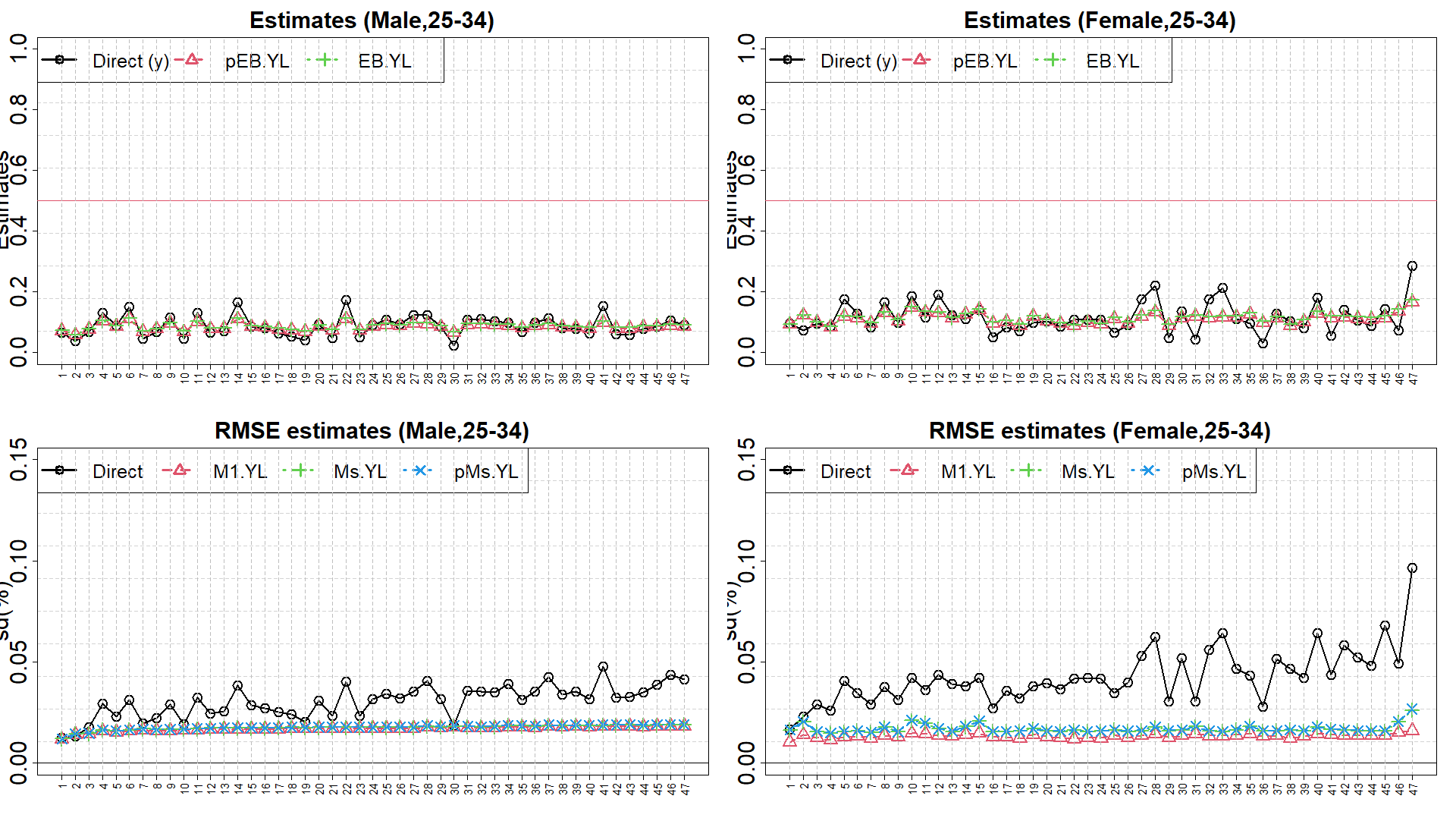}
    \caption{Estimates of poverty rate and these RMSE estimators for each prefecture}
    \label{Point_est}
\end{figure}

Moreover, we obtained the confidence intervals presented in Figure \ref{CIDA} for each prefecture, stratified by gender. The upper figure corresponded to the results for males, while the lower figure represented those for females.
We compared four types of confidence intervals: our transformed direct interval (TDirect), the explicit empirical Bayes (EB) interval (TEB.YL), the implicit bootstrap EB interval (TEB.B), and the naive interval (Mpnaive), for each prefecture.
The results indicated that TDirect and Mpnaive exhibit similar characteristics, whereas the two EB intervals (TEB.YL and TEB.B) also resemble each other. However, our two EB intervals are significantly shorter than others, particularly in cases with smaller sample sizes.

 \begin{figure}[h!]
    \centering
 \includegraphics[bb=0 0 1920 1082,width=10cm]{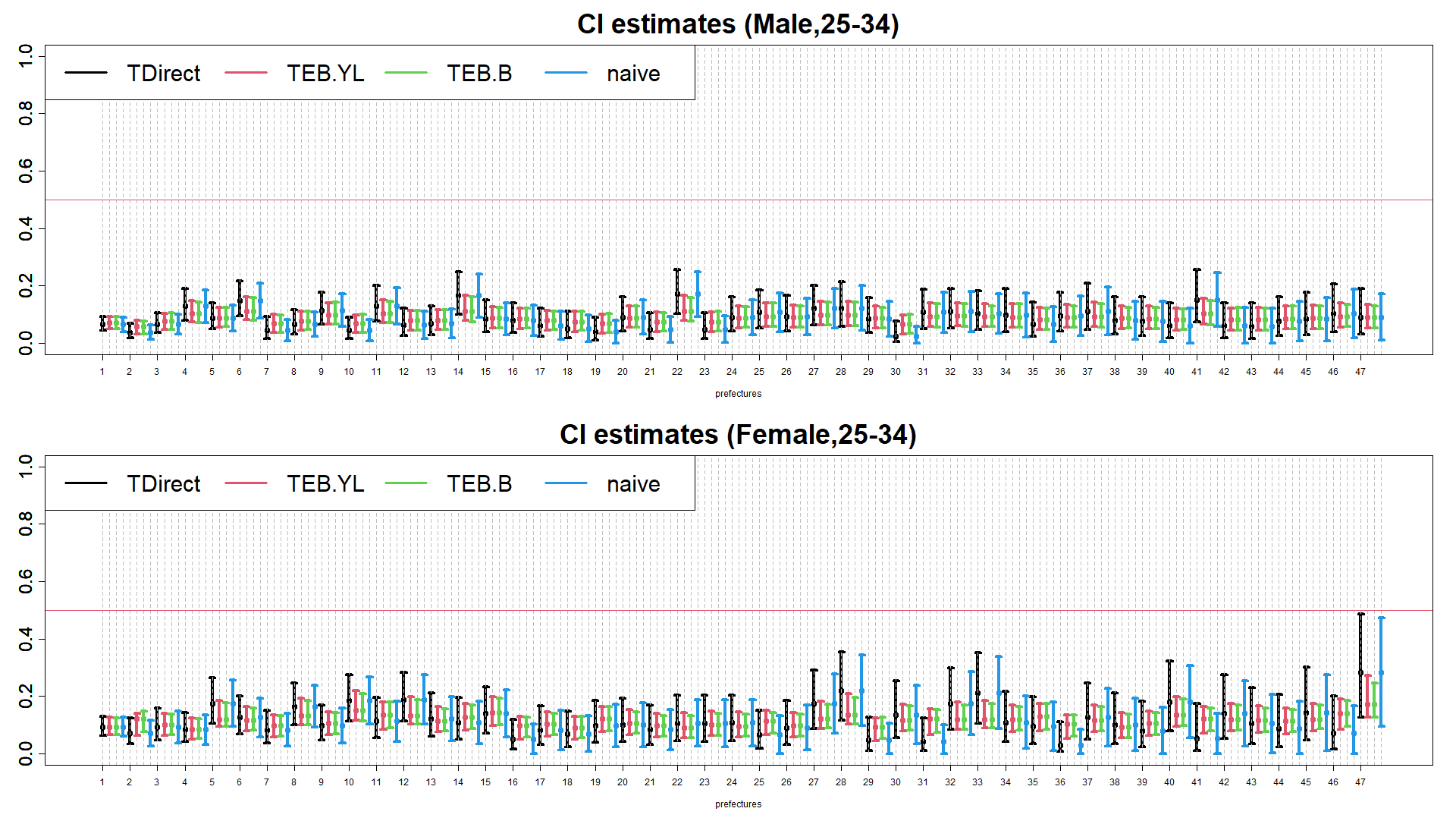}
 \caption{Confidence intervals for each prefecture}
 \label{CIDA}
 \end{figure}

Furthermore, we examined the ratio of the lengths of three confidence intervals relative to the transformed direct interval(TDirect). The results were presented in Figure \ref{LengthDA}.
The findings indicated that our two EB intervals achieve a substantial reduction in length compared to the transformed direct confidence interval.

\begin{figure}[h!]
    \centering
\includegraphics[bb=0 0 1920 1082,width=10cm]{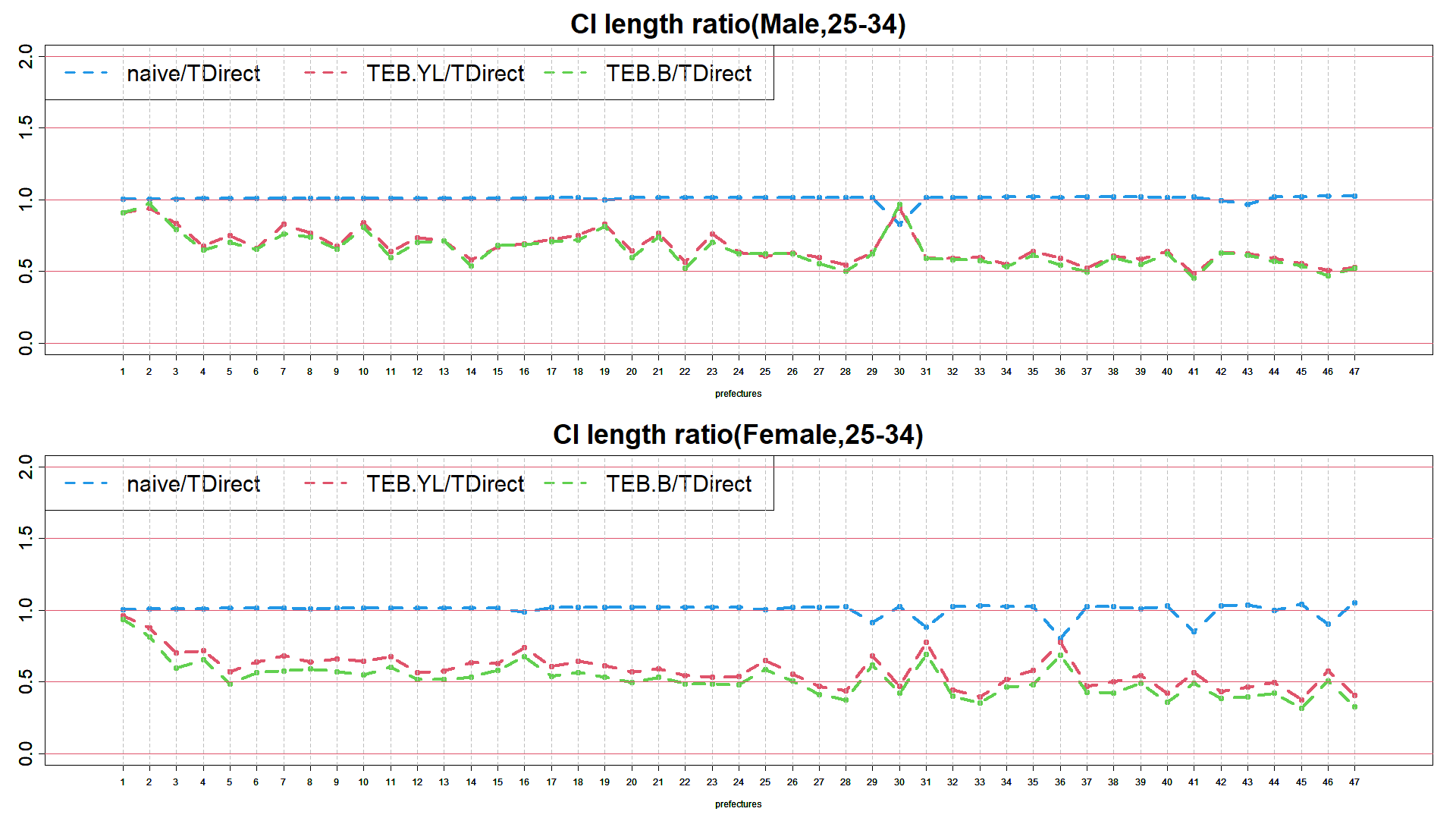}
\caption{Ratio of Length of each Confidence interval to that of TDirect}
\label{LengthDA}
\end{figure}

\section{Final Remarks}
The paper introduces an asymptotically bias-corrected variance stabilizing transformation in the context of area-level small area models. 
In addition, we consider a scenario where not just the number of small areas is large, but also some of these areas have large or moderately large sample sizes. Theoretical result was found in this setup, and are accompanied with data analysis and simulation. An important new research will be to extend this variance stabilizing idea to unit level models with or without moderately large samples in a given unit.

\section*{Acknowledgement}
This work was partially supported by the Japan Center for Economic Research and by JSPS KAKENHI Grant Number 22K01426. The analytical results presented here differ from the officially published statistics in Japan, as they were independently analyzed by the authors. We are grateful to the National Statistics Center, the Statistical Data Utilization Center, the on-site facilities at Okayama University, and ROIS-DS in Japan for allowing us to use microdata in secure facilities.

\appendix

\def\thesection{Appendix \Alph{section}}

\section{Proof of Theorem \ref{length.thm}}
The proof of Theorem \ref{length.thm} is given in this section. 

\begin{proof}
From Theorem \ref{th3}, the lengths are expressed as 
\begin{align*}
    L_{i,TD}&=\frac{1}{1+a_iD_i/2}\{g^{-1}(\hat I_{i,U}^D)-g^{-1}(\hat I_{i,L}^D)\},\\
    L_{i,TEB}&=\frac{1}{1+a_iD_i/2}\{g^{-1}(\hat I_{i,U}^{EB})-g^{-1}(\hat I_{i,L}^{EB})\},
\end{align*}
where $\hat I_{i,U}$ and $\hat I_{i,L}$ is the upper and lower limits of $\hat I_{i}$, respectively. 

\noindent Then it follows that
\begin{align}
    L_{i,TD}-L_{i,TEB}&=\frac{1}{1+a_iD_i/2}[g^{-1}(\hat I_{i,U}^D)-g^{-1}(\hat I_{i,L}^D)-\{g^{-1}(\hat I_{i,U}^{EB})-g^{-1}(\hat I_{i,L}^{EB})\}],\notag\\
    &=(1+O_p(n_i^{-1}))[g^{-1}(\hat I_{i,U}^D)-g^{-1}(\hat I_{i,L}^D)-\{g^{-1}(\hat I_{i,U}^{EB})-g^{-1}(\hat I_{i,L}^{EB})\}].\label{eq1250401}
\end{align}

\noindent We evaluated $g^{-1}(\hat I_{i,U}^D)-g^{-1}(\hat I_{i,L}^D)$ at first. 
\begin{align}
    g^{-1}(\hat I_{i,U}^D)&-g^{-1}(\hat I_{i,L}^D)=g^{-1}(\hat I_{i,U}^D)-g^{-1}(\hat I_{i,L}^{EB})-\{g^{-1}(\hat I_{i,L}^D)-g^{-1}(\hat I_{i,L}^{EB})\},\notag\\
    =&(\hat I_{i,U}^D-\hat I_{i,L}^D)\{g^{-1}(x)\}^{\prime}\Big |_{x=\hat I_{i,L}^{EB}}+\frac{1}{2}\{(\hat I_{i,U}^D-\hat I_{i,L}^{EB})^2-(\hat I_{i,L}^{D}-\hat I_{i,L}^{EB})^2\}\{g^{-1}(x)\}^{\prime \prime}\Big |_{x=\hat I_{i,L}^{EB}}\notag\\
     &+\frac{1}{6}\{(\hat I_{i,U}^D-\hat I_{i,L}^{EB})^3-(\hat I_{i,L}^{D}-\hat I_{i,L}^{EB})^3\}\{g^{-1}(x)\}^{\prime \prime \prime}\Big |_{x=x^*}\notag\\
     &+\frac{1}{6}(\hat I_{i,L}^{D}-\hat I_{i,L}^{EB})^3[\{g^{-1}(x)\}^{\prime \prime \prime}\Big |_{x=x^*}-\{g^{-1}(x)\}^{\prime \prime \prime}\Big |_{x=x^{**}}],\label{eq241018.1}
\end{align}
where $x^*$ and $x^{**}$ lies between $\hat I_{i,U}^{D}$ and $\hat I_{i,L}^{EB}$, $\hat I_{i,L}^{D}$ and $\hat I_{i,L}^{EB}$ respectively. 

\begin{align*}
    (\hat I_{i,U}^D-\hat I_{i,L}^{EB})^2-&(\hat I_{i,L}^{D}-\hat I_{i,L}^{EB})^2=(\hat I_{i,U}^D-\hat I_{i,L}^{D})^2+2(\hat I_{i,U}^D-\hat I_{i,L}^{D})(\hat I_{i,L}^{D}-\hat I_{i,L}^{EB}),\\
    =&4z_{\alpha/2}^2D_{i}+4z_{\alpha/2}\sqrt{D_{i}}\{g(y_i)-\hat \theta_i^{EB}-z_{\alpha/2}(\sqrt{D_i}-\sqrt{g_{1i}(\hat A)})\},\\
    =&4z_{\alpha/2}^2D_{i}+4z_{\alpha/2}\sqrt{D_{i}}\{g(y_i)-\hat \theta_i^{EB}\}+O_p(n_i^{-2}).
\end{align*}

\begin{align*}
    (\hat I_{i,U}^D-\hat I_{i,L}^{EB})^3-&(\hat I_{i,L}^{D}-\hat I_{i,L}^{EB})^3=(\hat I_{i,U}^D-\hat I_{i,L}^{D})^3+3(\hat I_{i,U}^D-\hat I_{i,L}^{D})^2(\hat I_{i,L}^{D}-\hat I_{i,L}^{EB})\\
    &+3(\hat I_{i,U}^D-\hat I_{i,L}^{D})(\hat I_{i,L}^{D}-\hat I_{i,L}^{EB})^2,\\
    =&8z_{\alpha/2}^3D_{i}^{3/2}+12z_{\alpha/2}^2D_{i}\{g(y_i)-\hat \theta_i^{EB}\}+6z_{\alpha/2}\sqrt{D_{i}}\{g(y_i)-\hat \theta_i^{EB}\}^2+O_p(n_i^{-2})
\end{align*}
In the above calculation, we used $\sqrt D_i-\sqrt{g_{1i}(\hat A)}=O_p(n_i^{-3/2})$ and $\hat I_{i,L}^{D}-\hat I_{i,L}^{EB}=g(y_i)-\hat \theta_i^{EB}+O_p(n_i^{-3/2})=O_p(n_i^{-1})$.

\begin{align*}
    \eqref{eq241018.1}=&2z_{\alpha/2}\sqrt{D_{i}}\{g^{-1}(x)\}^{\prime}\Big |_{x=\hat I_{i,L}^{EB}}
    +[2z_{\alpha/2}^2D_{i}+2z_{\alpha/2}\sqrt{D_{i}}\{g(y_i)-\hat \theta_i^{EB}\}]
    \{g^{-1}(x)\}^{\prime \prime}\Big |_{x=\hat I_{i,L}^{EB}}\\
    &+\left\{\frac{4}{3}z_{\alpha/2}^3D_{i}^{3/2}+2z_{\alpha/2}^2D_{i}\{g(y_i)-\hat \theta_i^{EB}\}+z_{\alpha/2}\sqrt{D_{i}}\{g(y_i)-\hat \theta_i^{EB}\}^2\right\}\{g^{-1}(x)\}^{\prime \prime \prime}\Big |_{x=x^*}\\
    &+O_p(n_i^{-2}),\\
    =&2z_{\alpha/2}\sqrt{D_{i}}\{g^{-1}(x)\}^{\prime}\Big |_{x=\hat I_{i,L}^{EB}}
    +2[z_{\alpha/2}^2D_{i}+z_{\alpha/2}\sqrt{D_{i}}\{g(y_i)-\hat \theta_i^{EB}\}
    \{g^{-1}(x)\}^{\prime \prime}\Big |_{x=\hat I_{i,L}^{EB}}\\
    &+\frac{4}{3}z_{\alpha/2}^3D_{i}^{3/2}\{g^{-1}(x)\}^{\prime \prime \prime}\Big |_{x=x^*}+O_p(n_i^{-2}).
\end{align*}
Note that $g(y_i)-\hat \theta_i^{EB}=O_p(n_i^{-1}).$

Next, using a similar way, it follows that 
\begin{align}
    g^{-1}(\hat I_{i,U}^{EB})&-g^{-1}(\hat I_{i,L}^{EB})
    =(\hat I_{i,U}^{EB}-\hat I_{i,L}^{EB})\{g^{-1}(x)\}^{\prime}\Big |_{x=\hat I_{i,L}^{EB}}+\frac{1}{2}(\hat I_{i,U}^{EB}-\hat I_{i,L}^{EB})^2\{g^{-1}(x)\}^{\prime \prime}\Big |_{x=\hat I_{i,L}^{EB}}\notag\\
    &+\frac{1}{6}(\hat I_{i,U}^{EB}-\hat I_{i,L}^{EB})^3\{g^{-1}(x)\}^{\prime \prime \prime}\Big |_{x=x^{***}},\notag\\
    =&2z_{\alpha/2}\sqrt{g_{1i}(\hat A)}\{g^{-1}(x)\}^{\prime}\Big |_{x=\hat I_{i,L}^{EB}}
    +2z_{\alpha/2}^2g_{1i}(\hat A)\{g^{-1}(x)\}^{\prime \prime}\Big |_{x=\hat I_{i,L}^{EB}}\notag\\
    &+\frac{4}{3}z_{\alpha/2}^3g_{1i}(\hat A)^{3/2}\{g^{-1}(x)\}^{\prime \prime \prime}\Big |_{x=x^{***}},    \label{eq241018.2}
\end{align}
where $x^{***}$ lies between $\hat I_{i,U}^{EB}$ and $\hat I_{i,L}^{EB}$. 

Hence, we obtain using \eqref{eq1250401}-\eqref{eq241018.2},
\begin{align*}
L_{i,TD}-L_{i,TEB}=&(1+O_p(n_i^{-1}))[g^{-1}(\hat I_{i,U}^D)-g^{-1}(\hat I_{i,L}^D)-\{g^{-1}(\hat I_{i,U}^{EB})-g^{-1}(\hat I_{i,L}^{EB})\}],\\
=&2z_{\alpha/2}(\sqrt{D_{i}}-\sqrt{g_{1i}(\hat A)})\{g^{-1}(x)\}^{\prime}\Big |_{x=\hat I_{i,L}^{EB}}\\
    &+2z_{\alpha/2}^2\{D_{i}-g_{1i}(\hat A)\}\{g^{-1}(x)\}^{\prime \prime}\Big |_{x=\hat I_{i,L}^{EB}}\\
    &+2z_{\alpha/2}\sqrt{D_{i}}\{g(y_i)-\hat \theta_i^{EB}\}
    \{g^{-1}(x)\}^{\prime \prime}\Big |_{x=\hat I_{i,L}^{EB}}\\
    &+\frac{4}{3}z_{\alpha/2}^3\{D_{i}^{3/2}-g_{1i}(\hat A)^{3/2})\}\{g^{-1}(x)\}^{\prime \prime \prime}\Big |_{x=x^*}\\
    &-\frac{4}{3}z_{\alpha/2}^3g_{1i}(\hat A)^{3/2}\left\{\{g^{-1}(x)\}^{\prime \prime \prime}\Big |_{x=x^{***}}-\{g^{-1}(x)\}^{\prime \prime \prime}\Big |_{x=x^{*}}
    \right\}+O_p(n_i^{-2}),\\
=&2z_{\alpha/2}\Big[(\sqrt{D_{i}}-\sqrt{g_{1i}(\hat A)})\{g^{-1}(x)\}^{\prime}\Big |_{x=\hat \theta_i^{EB}}\\
    &+\sqrt{D_{i}}\{g(y_i)-\hat \theta_i^{EB}\}
    \{g^{-1}(x)\}^{\prime \prime}\Big |_{x=\hat \theta_i^{EB}}\Big]+O_p(n_i^{-2}),
 \end{align*}
The last equation obtained by noting that $x^*-x^{***}=O_p(n_i^{-1})$ as follows from their definitions. The result that $\hat I_{i,L}^{EB}=\hat \theta_i^{EB}+O_p(n_i^{-1/2})$ is also used for the last equation. 

\end{proof}

\bibliography{PMP}
\bibliographystyle{acm} 
\bibliographystyle{abbrv}

\end{document}